\documentclass[11pt]{article}

\usepackage[utf8]{inputenc}
\usepackage[english]{babel}
\usepackage{amsmath,amsthm,amsfonts}
\usepackage[font=small]{caption}
\usepackage{float}
\usepackage{geometry}
\usepackage{graphicx}
\usepackage{hyperref}
\usepackage{indentfirst}
\usepackage{wrapfig}

\hypersetup{colorlinks}

\newtheorem{theorem}{Theorem}
\newtheorem{lemma}{Lemma}
\newtheorem{corollary}{Corollary}
\newtheorem{proposition}{Proposition}
\newtheorem{example}{Example}
\newtheorem{conjecture}{Conjecture}

\theoremstyle{definition}
\newtheorem{definition}{Definition}

\theoremstyle{remark}
\newtheorem{remark}{Remark}

\title{Feynman checkers: number-theoretic properties}

\author{F. Kuyanov\footnote{feodor.kuyanov@gmail.com; National Research University Higher School of Economics, Moscow, Russian Federation}\hspace{5pt} and A. Slizkov\footnote{elexunix@gmail.com; National Research University Higher School of Economics, Moscow, Russian Federation}}

\date{}

\begin{document}

\maketitle

\begin{abstract}
We study Feynman checkers, an elementary model of electron motion introduced by R. Feynman. In this model, a checker moves on a checkerboard, and we count the turns. Feynman checkers are also known as a one-dimensional quantum walk. We prove some new number-theoretic results in this model, for example, sign alternation of the real and imaginary parts of the electron wave function in a specific area. All our results can be stated in terms of Young diagrams, namely, we compare the number of Young diagrams with an odd and an even number of steps.
\end{abstract}



\section{Introduction}

In this work, we prove some new and easy-to-state results on Young diagrams, connected with the ``Feynman checkers'' model from quantum mechanics, introduced by R. Feynman \cite[Problem 2.6]{Feynman-Gibbs}.

\begin{wrapfigure}{r}{90pt}
    \centering
    \includegraphics[width=60pt]{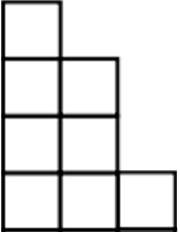}
    \caption{A Young diagram of size $3 \times 4$ with 3 steps.}
    \label{young-figure}
\end{wrapfigure}

Informally, a Young diagram is a set of checkered stripes aligned by their left sides, so that their lengths increase from top to bottom (see Figure \ref{young-figure}). The \emph{height} $h$ and \emph{width} $w$ are the number of stripes and the maximal stripe length respectively. \emph{The number of steps} (or \emph{outer corners}) in a Young diagram is the number of distinct stripe lengths. This paper is devoted to the following question: \emph{are there more Young diagrams of given size $w \times h$ with an odd or an even number of steps?} This question turns out to be highly nontrivial (see Figure \ref{a1-sign-alternation-figure}). In this figure, we observe different regimes near the middle and the sides of the angle.

Formally, a \emph{Young diagram of size $w \times h$} is a sequence of positive integers $x_1 \le \hdots \le x_h = w$; the \emph{number of steps} is the number of distinct integers among $x_1, \hdots, x_h$. The following two new theorems highlight both regimes in Figure \ref{a1-sign-alternation-figure}.

\begin{theorem}
\label{young-outside-theorem}
If $h/w > 3 + 2\sqrt{2}$, then the number of Young diagrams of size $w \times h$ with an odd number of steps exceeds the one with an even number of steps, if and only if $w$ is odd.
\end{theorem}

\begin{theorem}
\label{young-middle-theorem}
For any integer $d$ there exists $w_0$ such that for every $w > w_0$ the number of Young diagrams of size $w \times (w + d)$ with an odd number of steps exceeds the one with an even number of steps, if and only if $2w + d$ is $1$, $2$, $3$, or $4$ modulo $8$.
\end{theorem}

\begin{remark}
In Theorem \ref{young-outside-theorem}, the number $3 + 2\sqrt{2}$ is a sharp estimate by Proposition \ref{lowerbound-sharpness-proposition}.
\end{remark}

As a next step, it is natural to study the \emph{difference} between the number of Young diagrams of size $w \times h$ with an odd and an even number of steps. Let us discuss a few known observations. For $h = w$ even, the difference vanishes; for $h = w = 2n + 1$ odd, it is $(-1)^n \binom{2n}{n}$ (see Proposition \ref{middle-values-proposition} below). Such 4-periodicity roughly remains for $h$ close to $w$. For fixed half-perimeter $h + w$, the difference strongly oscillates as $h/w$ increases, attains a peak at $h/w \approx 3 + 2\sqrt{2}$, and then plummets to very small values (see Figure \ref{layer-a1-figure} and \cite[Corollary 2 and Theorems 2-4]{SU-22}). What is particularly notable, for some $h/w < 3 + 2\sqrt{2}$ the oscillation is weaker than in the vicinity, and such ``dips'' form a fractal structure for large $h + w$ (see Figure \ref{dips-a1-figure}). S. Nechaev (private communication) has posed the problem to find the positions of the ``dips''.

\begin{figure}[H]
    \centering
    \includegraphics[width=\textwidth]{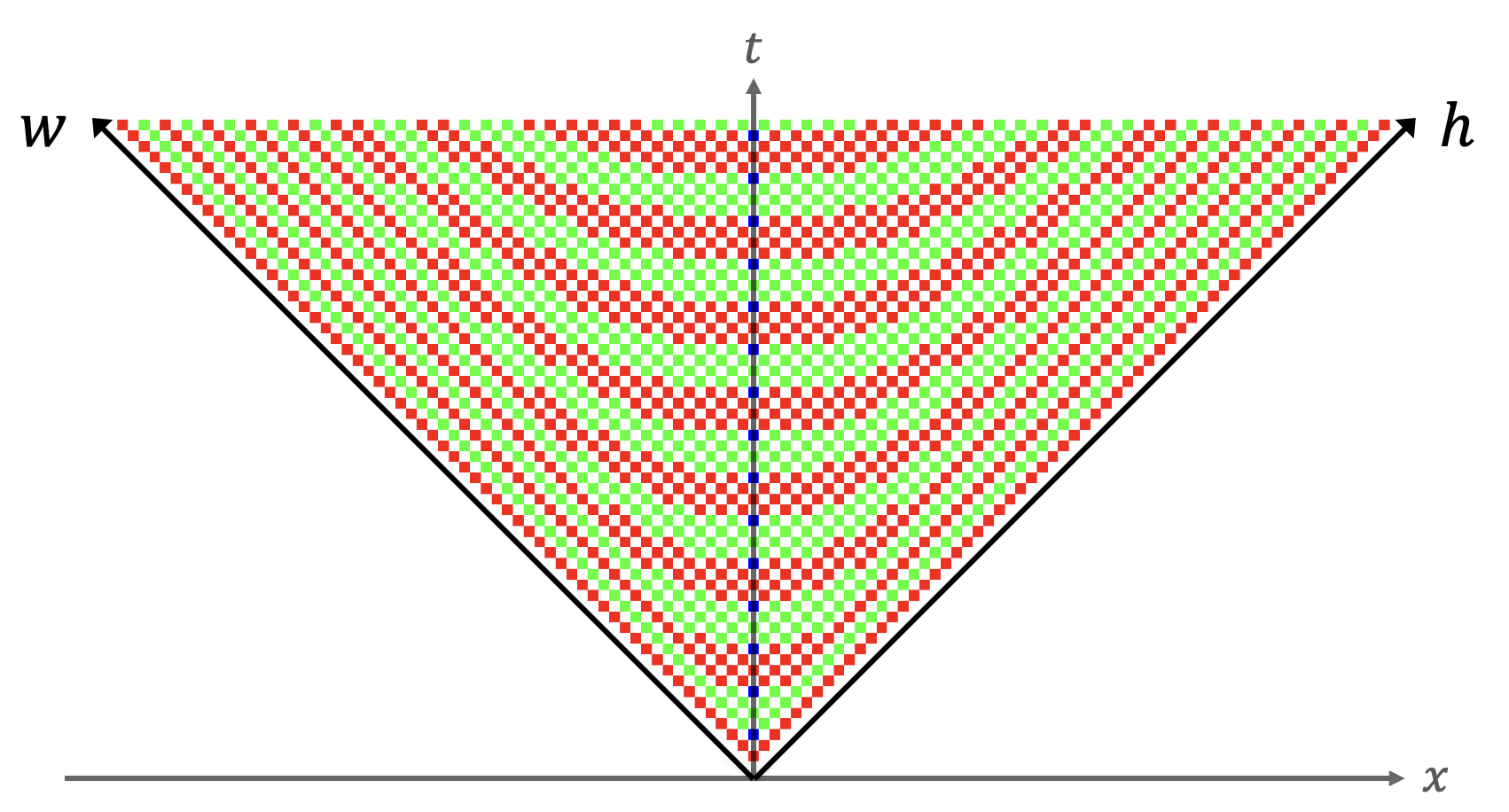}
    \caption{The sign of the difference between the number of Young diagrams of size $w \times h$ with an odd and an even number of steps. Red depicts $+1$, green depicts $-1$, blue depicts $0$.}
    \label{a1-sign-alternation-figure}
\end{figure}

\begin{figure}[H]
    \centering
    \includegraphics[width=\textwidth]{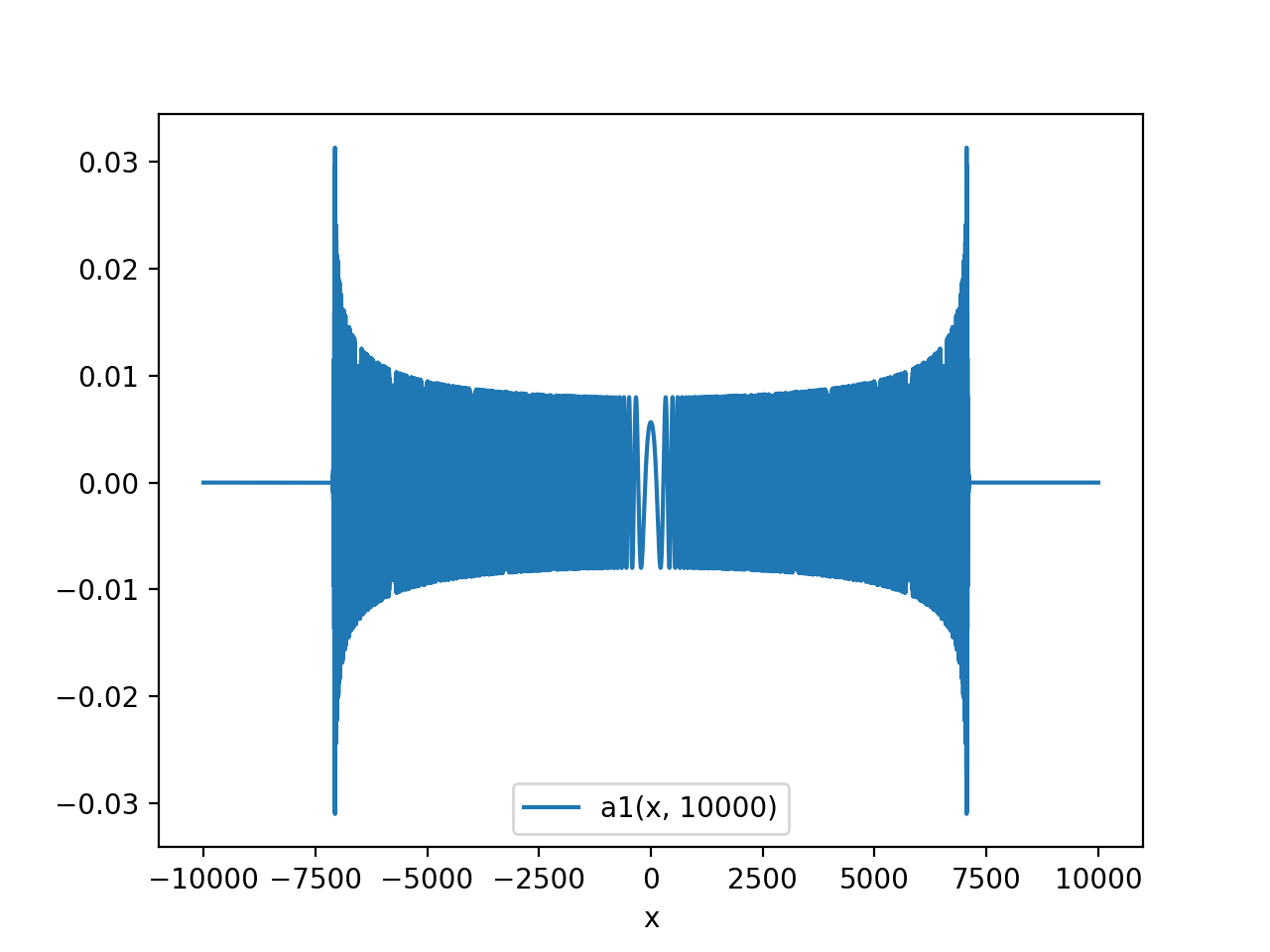}
    \caption{Normalized difference $\tilde a_1(x, t)$ between the number of Young diagrams of size $w \times h$ with an odd and an even number of steps, where $x = h - w$, $t = h + w - 1 = 10^4$.}
    \label{layer-a1-figure}
\end{figure}

\begin{figure}[H]
    \centering
    \includegraphics[width=\textwidth]{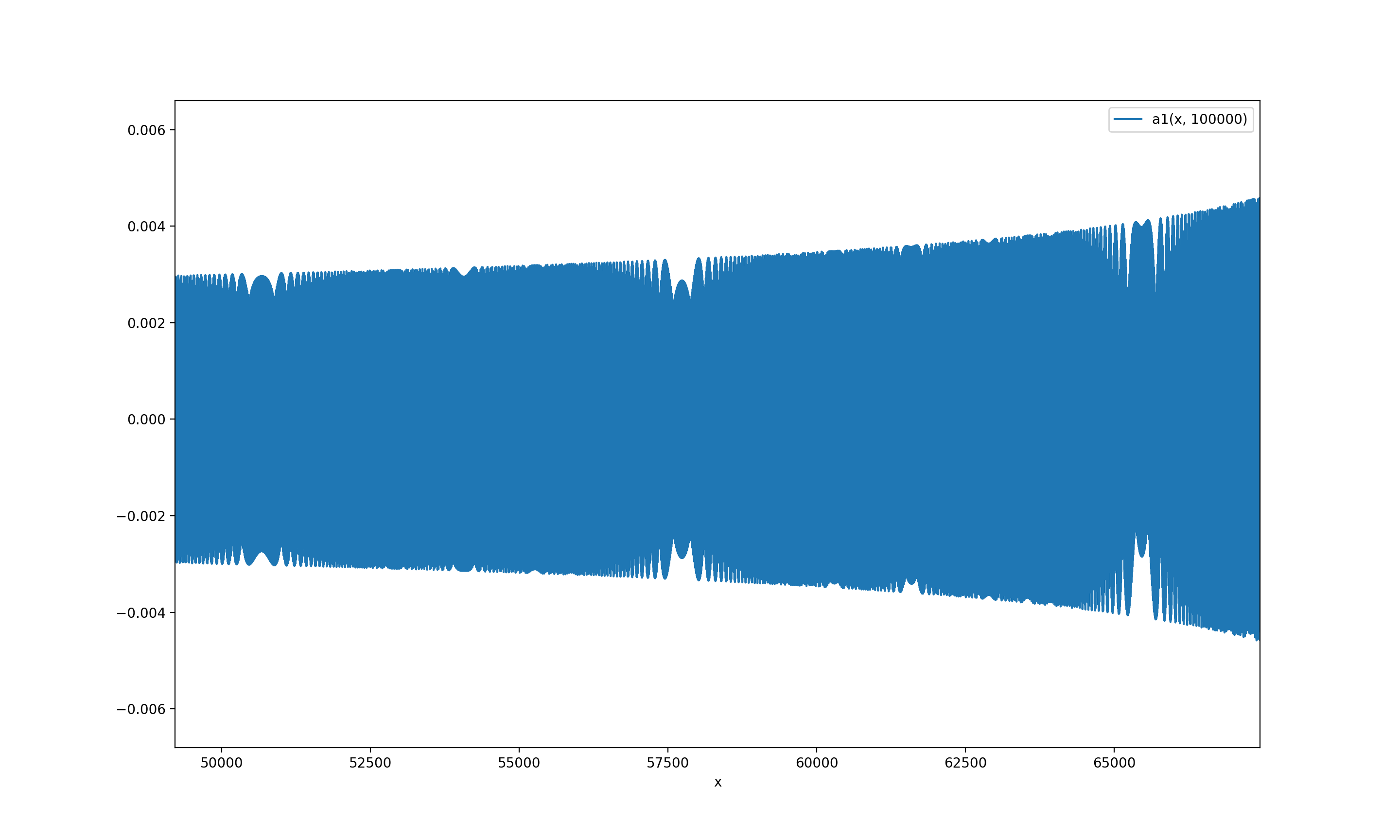}
    \caption{Fractal structure of the ``dips'' for $\tilde a_1(x, 10^5)$. Cf. the popcorn function \cite[Figure 1b]{Nechaev-Polovnikov-18}.}
    \label{dips-a1-figure}
\end{figure}

It turns out that this difference is proportional to the real part of the wave function in the simplest model of electron motion, known as Feynman checkers or a one-dimensional quantum walk (see \cite{Feynman-Gibbs}, \cite{SU-22}, and Definition \ref{checkers-definition}); this connection is made in Section \ref{background-section} after Definition \ref{checkers-definition}. In Section \ref{behaviour-outside-section} we prove sign alternation and damping of the real and imaginary parts of the wave function near the angle sides in Figure \ref{a1-sign-alternation-figure}, generalizing Theorem \ref{young-outside-theorem}. Our proof uses an equal-time recurrence relation (see Proposition \ref{recurrence-proposition}), which has recently been obtained in \cite{SU-22}. Then, in Section \ref{behaviour-middle-section} we prove Theorem \ref{young-middle-theorem}, describing behaviour of the difference near the middle of the angle. In addition to the recurrence relation, the proof uses a known asymptotic formula for the wave function (see Proposition \ref{asymptotic-a1-proposition}). The leading terms of the asymptotic formula have been found in \cite[Theorem 2]{Ambainis-etal-01}, and the remainder terms have been estimated in \cite[Proposition 2.2]{Sunada-Tate-12} and \cite[Theorem 2]{SU-22}. See also \cite{Zakorko-22} where a stronger asymptotic formula has been obtained. Section \ref{dips-section} is devoted to the dips-position problem posed by S. Nechaev. We introduce a precise definition of dips and obtain an explicit formula for their positions (see Theorem \ref{dips-a1-theorem}). This formula has a remarkable physical interpretation: the dips are caused by electron diffraction on the integer lattice and occur for those electron velocities which correspond to a rational de Broglie wavelength (cf. \cite{Nechaev-Polovnikov-18}). Finally, we prove the sharpness of the lower bounds in Theorems \ref{oscillation-outside-theorem} and \ref{young-outside-theorem}.

\section{Background}
\label{background-section}

Let us give the definition of Feynman checkers. See \cite{Dmitriev-22,Ozhegov-22,SU-22} for generalizations.

\begin{definition}
\label{checkers-definition}
(see \cite[Definition 2]{SU-22})
Fix $\varepsilon > 0$ and $m \ge 0$ called \emph{lattice step} and \emph{particle mass} respectively. Consider the lattice $\varepsilon \mathbb{Z}^2 = \{(x, t) : x/\varepsilon, t/\varepsilon \in \mathbb{Z}\}$. A \emph{checker path} $s$ is a finite sequence of lattice points such that the vector from each point (except the last one) to the next one equals either $(\varepsilon, \varepsilon)$ or $(-\varepsilon, \varepsilon)$. Denote by $\mathrm{turns}(s)$ the number of points in $s$ (not the first or the last one) such that the vectors from the point to the next and the previous ones are orthogonal. For each $(x, t) \in \varepsilon \mathbb{Z}^2$, where $t > 0$, the \emph{wave function} is
\begin{align}
    \label{a-definition}
    a(x, t, m, \varepsilon) &:= (1 + m^2\varepsilon^2)^{(1 - t/\varepsilon)/2} \, i \, \sum_s (-im\varepsilon)^{\mathrm{turns}(s)},
\end{align}
where sum is over all checker paths $s$ from $(0, 0)$ to $(x, t)$ containing $(\varepsilon, \varepsilon)$. Denote
\begin{align*}
	P(x, t, m, \varepsilon) &:= |a(x, t, m, \varepsilon)|^2, \\
	\tilde a_1(x, t, m, \varepsilon) &:= \operatorname{Re} a(x, t + \varepsilon, m, \varepsilon), \\
	\tilde a_2(x, t, m, \varepsilon) &:= \operatorname{Im} a(x + \varepsilon, t + \varepsilon, m, \varepsilon).
\end{align*}
Hereafter we omit the argument $\varepsilon$, if $\varepsilon = 1$; we omit both arguments $m$ and $\varepsilon$, if $m = \varepsilon = 1$.
\end{definition}

One interprets $P(x, t, m, \varepsilon)$ as the probability to find an electron of mass $m$ in the square $\varepsilon \times \varepsilon$ with the center $(x, t)$, if the electron has been emitted from the origin. Notice that the value $m\varepsilon$, hence $P(x, t, m, \varepsilon)$, is dimensionless in the natural units, where $\hbar = c = 1$.

\begin{wrapfigure}{r}{180pt}
    \centering
    \includegraphics[width=150pt]{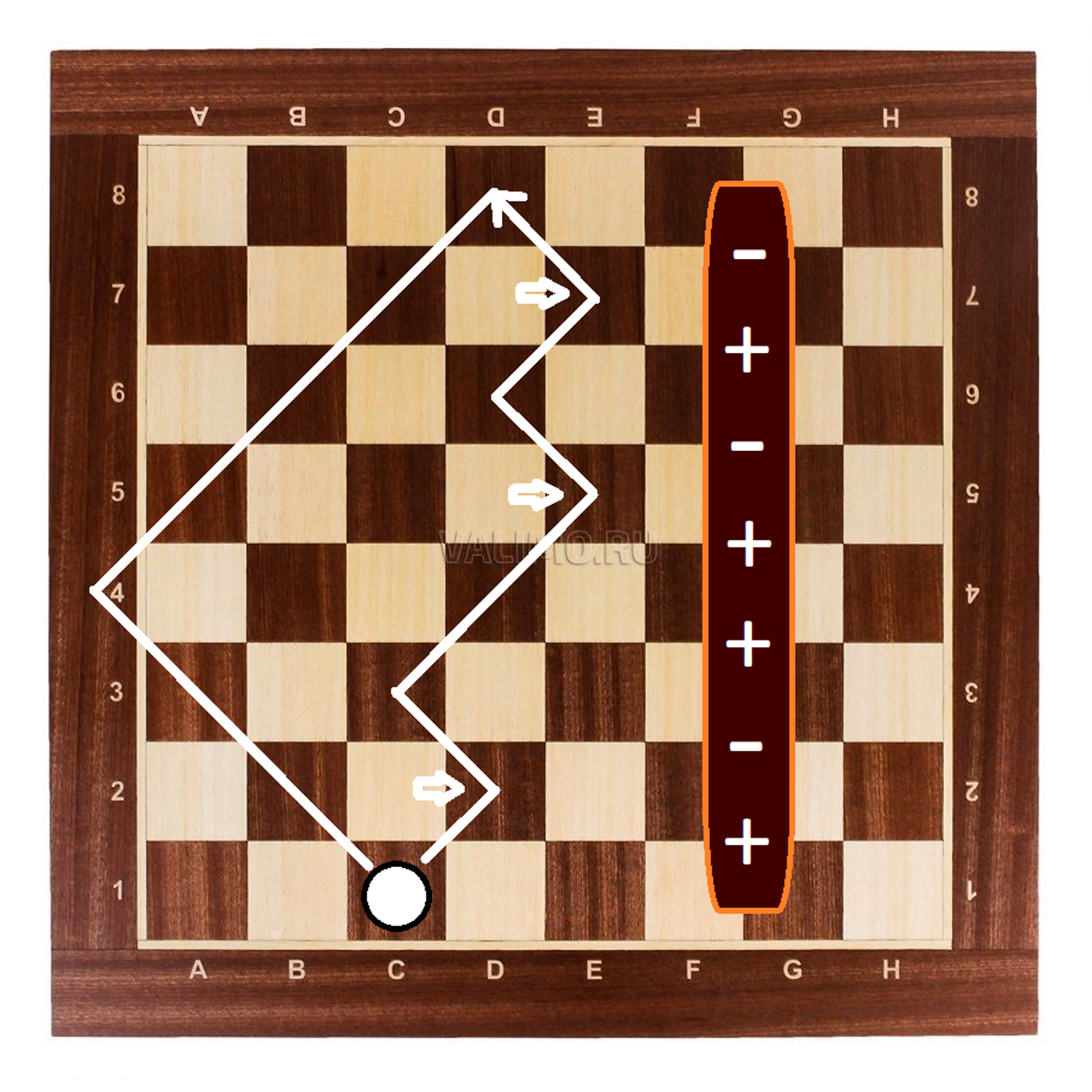}
    \caption{\cite{SU-22} A checker path from $(0, 0)$ to $(1, 7)$ with 5 turns and the corresponding Young diagram.}
    \label{checker-path-figure}
\end{wrapfigure}

Let us make the connection between Feynman checkers for $m = \varepsilon = 1$ and Young diagrams. As in Figure \ref{checker-path-figure}, to a path $s$ from $(0, 0)$ to $(x, t)$ through $(1, 1)$ with an odd number of turns assign the Young diagram obtained by drawing the lines from $(0, 0)$ and $(x, t)$ upwards-left and downwards-left respectively to their intersection point and rotating through $45^\circ$ counterclockwise. The value $\mathrm{turns}(s)$ modulo 4 affects both the parity of the number of steps in the resulting Young diagram and the sign of the corresponding term in \eqref{a-definition}. For instance, in Figure \ref{checker-path-figure}, $\mathrm{turns}(s)$ is 1 modulo 4, the number of steps is odd, and the sign of the corresponding term is ``+''. Thus the value $2^{(t-1)/2}\,\operatorname{Re} a(x, t)$ can be interpreted as the difference between the number of Young diagrams of size $\frac{t - x}{2} \times \frac{t + x}{2}$ with an odd and an even number of steps respectively.

\begin{example}[Boundary values]
\label{boundary-values-example}
\cite[Example 2]{SU-22}
For each $t \in \varepsilon \mathbb{Z}$, where $t > 0$, we have 
\begin{align*}
	a(t, t, m, \varepsilon) &= i(1 + m^2\varepsilon^2)^{(1 - t/\varepsilon)/2}, \\ 
	a(2\varepsilon - t, t, m, \varepsilon) &= m\varepsilon(1 + m^2\varepsilon^2)^{(1 - t/\varepsilon)/2}, \\
	a(t - 2\varepsilon, t, m, \varepsilon) &= (m\varepsilon + i(2 - t/\varepsilon)m^2\varepsilon^2)(1 + m^2\varepsilon^2)^{(1 - t/\varepsilon)/2},
\end{align*}
and for each $x > t$ or $x \le -t$ we have $a(x, t, m, \varepsilon) = 0$.
\end{example}

\begin{remark}
\label{eps-normalization-remark}
$a(x, t, m, \varepsilon) = a(x/\varepsilon, t/\varepsilon, m\varepsilon, 1)$.
\end{remark}

Now we state several known results to be used in our proofs below.

\begin{proposition}[Middle values]
\label{middle-values-proposition}
\cite[Proposition 4 and 18(B)]{SU-22}
For each $0 \le k < t$ the number $\tilde a_1(-t + 2k + 1, t)$ is the coefficient before $z^{t-k-1}$ in the expansion of the polynomial $2^{-t/2}\,(1 + z)^{t-k-1}\,(1 - z)^k$. In particular,
\begin{align*}
    \tilde a_1(0, 4n + 1) &= \frac{(-1)^n}{2^{(4n+1)/2}} \binom{2n}{n}, & \tilde a_1(0, 4n + 3) &= 0, \\
    \tilde a_1(2, 4n + 1) &= \frac{(-1)^n}{2^{(4n-1)/2}} \binom{2n - 1}{n}, & \tilde a_1(2, 4n + 3) &= \frac{(-1)^n}{2^{(4n+3)/2}} \left(\binom{2n}{n} - \binom{2n}{n + 1}\right).
\end{align*}
\end{proposition}

The latter two formulae are the particular case of \cite[Proposition 18(B)]{SU-22} for $x = 2$.

\begin{proposition}[Symmetry]
\label{a1-symmetry-proposition}
\cite[Proposition 8]{SU-22}
For each $(x, t) \in \varepsilon \mathbb{Z}^2$, where $t \ge 0$, we have
\begin{align*}
    \tilde a_1(x, t, m, \varepsilon) = \tilde a_1(-x, t, m, \varepsilon).
\end{align*}
\end{proposition}

\begin{proposition}[Equal-time recurrence relation]
\label{recurrence-proposition}
\cite[Proposition 10]{SU-22}, \cite[Proposition 15]{SU2-22}
For each $(x, t) \in \varepsilon \mathbb{Z}^2$, where $t > 0$, we have
\begin{align*}
	(x + \varepsilon)((x - \varepsilon)^2 - t^2)\,\tilde a_1(x - 2\varepsilon, t, m, \varepsilon) + (x - \varepsilon)((x + \varepsilon)^2 - t^2)\,\tilde a_1(x + 2\varepsilon, t, m, \varepsilon) = \\
	= 2x((1 + 2m^2\varepsilon^2)(x^2 - \varepsilon^2) - t^2)\,\tilde a_1(x, t, m, \varepsilon), \\
	(x + \varepsilon)((x - \varepsilon)^2 - (t + \varepsilon)^2)\,\tilde a_2(x - 2\varepsilon, t, m, \varepsilon) + (x - \varepsilon)((x + \varepsilon)^2 - (t - \varepsilon)^2)\,\tilde a_2(x + 2\varepsilon, t, m, \varepsilon) = \\
	= 2x((1 + 2m^2\varepsilon^2)(x^2 - \varepsilon^2) - t^2 + \varepsilon^2)\,\tilde a_2(x, t, m, \varepsilon).
\end{align*}
\end{proposition}

\begin{proposition}[Large-time asymptotic formula between the peaks]
\label{asymptotic-a1-proposition}
\cite[Theorem 2]{SU-22}
For each $\delta > 0$ there is $C_\delta > 0$ such that for each $m, \varepsilon > 0$ and $(x, t) \in \varepsilon \mathbb{Z}^2$ with $(x + t) / \varepsilon$ odd satisfying
\begin{align*}
	|x|/t < 1/\sqrt{1 + m^2\varepsilon^2} - \delta, \qquad \varepsilon \le 1/m, \qquad t > C_\delta/m,
\end{align*}
we have
\begin{align}
	\tilde a_1(x, t, m, \varepsilon) &= \varepsilon\sqrt{\frac{2m}{\pi}}\,(t^2 - (1 + m^2\varepsilon^2)x^2)^{-1/4}\,\sin \theta(x, t, m, \varepsilon) + O_\delta\left(\frac{\varepsilon}{m^{1/2}t^{3/2}}\right),
	\intertext{where}
	\label{theta-definition}
	\theta(x, t, m, \varepsilon) &:= \frac{t}{\varepsilon}\left(\arcsin \frac{m\varepsilon}{\sqrt{(1 + m^2\varepsilon^2)(1 - (x/t)^2)}} - \frac{x}{t} \arcsin \frac{m\varepsilon x/t}{\sqrt{1 - (x/t)^2}}\right) + \frac{\pi}{4}.
\end{align}
\end{proposition}

Hereafter notation $f(x, t, m, \varepsilon) = O_\delta(g(x, t, m, \varepsilon))$ means that there is a constant $C(\delta)$ (depending on $\delta$ but \emph{not} on $x, t, m, \varepsilon$) such that for each $x, t, m, \varepsilon, \delta$ satisfying the assumptions of the theorem we have $|f(x, t, m, \varepsilon)| \le C(\delta)\left|g(x, t, m, \varepsilon)\right|$. From now on we omit the arguments $m, \varepsilon$ of the function $\theta(x, t, m, \varepsilon)$ in the same way as we do in Definition \ref{checkers-definition}.

\section{Behaviour near the angle sides}
\label{behaviour-outside-section}

The following theorem explains sign alternation in Figure \ref{a1-sign-alternation-figure} below the lines $t = \pm \sqrt{2}x$ and generalizes Theorem \ref{young-outside-theorem}.

\begin{theorem}
\label{oscillation-outside-theorem}
For each $k \in \{1, 2\}$, $m > 0$, and $(x, t) \in \varepsilon \mathbb{Z}^2$ such that $t > 0$ and $\frac{1}{\sqrt{1 + m^2\varepsilon^2}} \le \frac{x}{t} \le 1$ we have
\begin{align*}
	\mathrm{sgn}(\tilde a_k(x, t, m, \varepsilon)) &= (-1)^{\frac{t - x + k\varepsilon}{2\varepsilon} - 1}, \\
	|\tilde a_k(x - \varepsilon, t, m, \varepsilon)| &> |\tilde a_k(x + \varepsilon, t, m, \varepsilon)|,
\end{align*}
for even and odd $\frac{x + t}{\varepsilon} + k$, respectively.
\end{theorem}

\begin{proof}[Proof of Theorem \ref{oscillation-outside-theorem}]
Fix $t$, $m$, $\varepsilon$ and denote
\begin{align*}
	b_k(x) &:= (-1)^{\frac{t - x + k\varepsilon}{2\varepsilon} - 1}\,\tilde a_k(x, t, m, \varepsilon).
\end{align*}
It suffices to prove that
\begin{align*}
	b_k(x - 2\varepsilon) &> b_k(x) \ge 0 \quad\text{for}\quad \frac{t}{\sqrt{1 + m^2\varepsilon^2}} + \varepsilon < x \le t + \varepsilon.
\end{align*}
For $k = 1$, we prove it by induction on $x$ with step $2\varepsilon$ in the descending order. The base ($x = t + \varepsilon$) follows from the first and third equations in Example \ref{boundary-values-example}. The induction step is obtained from the following chain of relations:
\begin{align*}
	b_1(x - 2\varepsilon) = \frac{2x((1 + 2m^2\varepsilon^2)(x^2 - \varepsilon^2) - t^2)\,b_1(x) - (x - \varepsilon)(t^2 - (x + \varepsilon)^2)\,b_1(x + 2\varepsilon)}{(x + \varepsilon)(t^2 - (x - \varepsilon)^2)} \ge
\end{align*}
\begin{align*}
	\ge \frac{2x((1 + 2m^2\varepsilon^2)(x^2 - \varepsilon^2) - t^2) - (x - \varepsilon)(t^2 - (x + \varepsilon)^2)}{(x + \varepsilon)(t^2 - (x - \varepsilon)^2)}\,b_1(x) > b_1(x).
\end{align*}
Here the first equality follows from Proposition \ref{recurrence-proposition}. The next inequality holds by the inductive hypothesis and the non-negativity of $t^2 - (x + \varepsilon)^2$. Now let us check the last inequality. Since $t^2 - (x - \varepsilon)^2 > 0$, it is equivalent to
\begin{align*}
	2x((1 + 2m^2\varepsilon^2)(x^2 - \varepsilon^2) - t^2) - (x - \varepsilon)(t^2 - (x + \varepsilon)^2) > (x + \varepsilon)(t^2 - (x - \varepsilon)^2).
\end{align*}
After expansion and simplification using that $x > 0$ we get
\begin{align*}
	t^2 < (1 + m^2\varepsilon^2)(x^2 - \varepsilon^2).
\end{align*}
The resulting inequality is equivalent to $x^2 > \frac{t^2}{1 + m^2\varepsilon^2} + \varepsilon^2$, which can be obtained from the assumption $x \ge \frac{t}{\sqrt{1 + m^2\varepsilon^2}} + \varepsilon$ by squaring.
\\[10pt]
Let us prove $b_2(x - 2\varepsilon) > b_2(x)$ similarly. The induction base ($x = t$) follows from the first and third equations in Example \ref{boundary-values-example} by applying the following estimate:
\begin{align*}
	\frac{t}{\varepsilon} > \frac{\sqrt{1 + m^2\varepsilon^2}}{\sqrt{1 + m^2\varepsilon^2} - 1} = 1 + \frac{1}{\sqrt{1 + m^2\varepsilon^2} - 1} > 1 + \frac{1}{1 + m^2\varepsilon^2 - 1} = 1 + \frac{1}{m^2\varepsilon^2},
\end{align*}
because $x = t > \frac{t}{\sqrt{1 + m^2\varepsilon^2}} + \varepsilon$. The induction step is obtained from the following chain of relations:
\begin{align*}
	b_2(x - 2\varepsilon) = \frac{2x((1 + 2m^2\varepsilon^2)(x^2 - \varepsilon^2) - t^2 + \varepsilon^2)\,b_2(x) - (x - \varepsilon)((t - \varepsilon)^2 - (x + \varepsilon)^2)\,b_2(x + 2\varepsilon)}{(x + \varepsilon)((t + \varepsilon)^2 - (x - \varepsilon)^2)} \ge
\end{align*}
\begin{align*}
	\ge \frac{2x((1 + 2m^2\varepsilon^2)(x^2 - \varepsilon^2) - t^2 + \varepsilon^2) - (x - \varepsilon)((t - \varepsilon)^2 - (x + \varepsilon)^2)}{(x + \varepsilon)((t + \varepsilon)^2 - (x - \varepsilon)^2)}\,b_2(x) > b_2(x).
\end{align*}
Here the first equality follows from Proposition \ref{recurrence-proposition}. The next inequality holds by the inductive hypothesis and the non-negativity of $(t - \varepsilon)^2 - (x + \varepsilon)^2$. Now let us check the last inequality. Since $(t + \varepsilon)^2 - (x - \varepsilon)^2 > 0$, it is equivalent to
\begin{align*}
	2x((1 + 2m^2\varepsilon^2)(x^2 - \varepsilon^2) - t^2 + \varepsilon^2) - (x - \varepsilon)((t - \varepsilon)^2 - (x + \varepsilon)^2) > (x + \varepsilon)((t + \varepsilon)^2 - (x - \varepsilon)^2).
\end{align*}
After expansion and simplification using that $x > 0$ we get
\begin{align*}
	t^2 + \frac{\varepsilon^2t}{x} < (1 + m^2\varepsilon^2)(x^2 - \varepsilon^2).
\end{align*}
The resulting inequality is equivalent to $x^2 > \frac{t^2 + \varepsilon^2t/x}{1 + m^2\varepsilon^2} + \varepsilon^2$, which can be obtained from the assumption $x \ge \frac{t}{\sqrt{1 + m^2\varepsilon^2}} + \varepsilon$ by squaring using the following estimate:
\begin{align*}
    \frac{\varepsilon^2t/x}{1 + m^2\varepsilon^2} < \frac{\varepsilon^2\sqrt{1 + m^2\varepsilon^2}}{1 + m^2\varepsilon^2} = \frac{\varepsilon^2}{\sqrt{1 + m^2\varepsilon^2}} \le \frac{2\varepsilon t}{\sqrt{1 + m^2\varepsilon^2}}.
\end{align*}
\end{proof}

\begin{proof}[Proof of Theorem \ref{young-outside-theorem}]
It follows immediately from Theorem \ref{oscillation-outside-theorem} by substitution $m = \varepsilon = 1$, $x = h - w$, $t = h + w - 1$, because the assumptions of Theorem \ref{oscillation-outside-theorem} are obtained from the following estimate:
\begin{align*}
	\frac{x}{t} = \frac{h - w}{h + w - 1} > \frac{h - w}{h + w} > \frac{3 + 2\sqrt{2} - 1}{3 + 2\sqrt{2} + 1} = \frac{1}{\sqrt{2}}.
\end{align*}
\end{proof}

\section{Behaviour near the angle middle}
\label{behaviour-middle-section}

Let us restate Theorem \ref{young-middle-theorem} in terms of Feynman checkers.

\begin{theorem}
\label{oscillation-middle-theorem}
For each fixed integer $x \neq 0$, all sufficiently large integers $t \not\equiv x \bmod 2$ satisfy
\begin{align*}
    \mathrm{sgn}(\tilde a_1(x, t)) = (-1)^{\left\lfloor t/4 \right\rfloor}.
\end{align*}
\end{theorem}

\begin{proof}[Proof of Theorem \ref{oscillation-middle-theorem}]
First we prove the theorem in the case $t \not\equiv 3 \bmod 4$. According to Proposition \ref{asymptotic-a1-proposition}, for fixed $x$ and sufficiently large $t \not\equiv x \bmod 2$ we have
\begin{align*}
    \tilde a_1(x, t) = \sqrt{\frac{2}{\pi}} \frac{\sin \theta(x, t)}{\sqrt[4]{t^2 - 2x^2}} + O\left(\frac{1}{t^{3/2}}\right),
\end{align*}
and
\begin{align*}
    \theta(x, t) = t \arcsin \frac{t}{\sqrt{2(t^2 - x^2)}} - x \arcsin \frac{x}{\sqrt{t^2 - x^2}} + \frac{\pi}{4} = \\
    = t \arcsin \left(\frac{1}{\sqrt{2}} + O\left(\frac{1}{t^2}\right)\right) - x \arcsin O\left(\frac{1}{t}\right) + \frac{\pi}{4} = \frac{\pi}{4}(t + 1) + o(1) \quad\text{as $t \to \infty$}.
\end{align*}
Thus, for any fixed $x$ and for all sufficiently large $t \not\equiv 3 \bmod 4$ such that $x + t$ is odd we have
\begin{align*}
    \mathrm{sgn}(\tilde a_1(x, t)) = \mathrm{sgn}\left(\sin\left(\frac{\pi}{4}(t + 1)\right)\right) = (-1)^{\left\lfloor t/4 \right\rfloor}. 
\end{align*}

Now suppose $t \equiv 3 \bmod 4$. The case $x = 2$ follows from the last equation in Proposition \ref{middle-values-proposition}. It suffices to prove the theorem in the case $x > 2$ by Proposition \ref{a1-symmetry-proposition}. Lemma \ref{neighbour-relation-lemma} below implies that eventually $\frac{\tilde a_1(x, t)}{\tilde a_1(x - 2, t)} > 1$, eventually $\frac{\tilde a_1(x - 2, t)}{\tilde a_1(x - 4, t)} > 1$, and so on. Multiplying the inequalities we get $\frac{\tilde a_1(x, t)}{\tilde a_1(2, t)} > 1$, in particular, $\mathrm{sgn}(\tilde a_1(x, t)) = \mathrm{sgn}(\tilde a_1(2, t)) = (-1)^{\left\lfloor t/4 \right\rfloor}$.
\end{proof}

\begin{lemma}
\label{neighbour-relation-lemma}
For each fixed even $x \ge 2$ there exists $\lim_{\substack{t \to \infty \\ t \equiv_4 3}} \frac{\tilde a_1(x + 2, t)}{\tilde a_1(x, t)} > 1$.
\end{lemma}

\begin{proof}
We prove the lemma by induction on $x \ge 2$. Denote $k(x) := \lim_{\substack{t \to \infty \\ t \equiv_4 3}} \frac{\tilde a_1(x + 2, t)}{\tilde a_1(x, t)}$. The base ($x = 2$) follows from Proposition \ref{recurrence-proposition} for $m = \varepsilon = 1$:
\begin{align*}
    k(2) = \lim_{\substack{t \to \infty \\ t \equiv_4 3}} \frac{\tilde a_1(4, t)}{\tilde a_1(2, t)} = \lim_{\substack{t \to \infty \\ t \equiv_4 3}} \frac{2x(3(x^2 - 1) - t^2)\,\tilde a_1(x, t) - (x + 1)((x - 1)^2 - t^2)\,\tilde a_1(x - 2, t)}{(x - 1)((x + 1)^2 - t^2)\,\tilde a_1(x, t)} \bigg\rvert_{x = 2} = \\
    = \lim_{\substack{t \to \infty \\ t \equiv_4 3}} \frac{4(9 - t^2)\,\tilde a_1(2, t) - 3(1 - t^2)\,\tilde a_1(0, t)}{(9 - t^2)\,\tilde a_1(2, t)} = 4,
\end{align*}
because $\tilde a_1(0, 4n + 3) = 0$ and $\tilde a_1(2, 4n + 3) \neq 0$ by Proposition \ref{middle-values-proposition}. The induction step follows from
\begin{align*}
    k(x) = \lim_{\substack{t \to \infty \\ t \equiv_4 3}} \frac{\tilde a_1(x + 2, t)}{\tilde a_1(x, t)} = \lim_{\substack{t \to \infty \\ t \equiv_4 3}} \frac{2x(3(x^2 - 1) - t^2)\,\tilde a_1(x, t) - (x + 1)((x - 1)^2 - t^2)\,\tilde a_1(x - 2, t)}{(x - 1)((x + 1)^2 - t^2)\,\tilde a_1(x, t)} = \\
    = \lim_{u \to \infty} \frac{2x(u - 3(x^2 - 1))\,k(x - 2) - (x + 1)(u - (x - 1)^2)}{(x - 1)(u - (x + 1)^2)\,k(x - 2)} = \frac{2xk(x - 2) - x - 1}{(x - 1)\,k(x - 2)} > 1,
\end{align*}
because $k(x - 2) > 1$ by the induction hypothesis and $x > 2$.
\end{proof}

\begin{proof}[Proof of Theorem \ref{young-middle-theorem}]
It follows trivially from Theorem \ref{oscillation-middle-theorem} by substitution $x = d$, $t = 2w + d - 1$.
\end{proof}

\begin{corollary}
\label{a1-nonzero-corollary}
For each fixed integer $x \neq 0$, all sufficiently large integers $t \not\equiv x \bmod 2$ satisfy $\tilde a_1(x, t) \neq 0$.
\end{corollary}

\section{Dips positions}
\label{dips-section}

Let us give a precise definition of dips, which describes the positions of the minimal oscillation of the function $\tilde a_1(x, t, m, \varepsilon)$ for fixed $t, m, \varepsilon$ (see Figure \ref{dips-a1-figure}). For simplicity, we first give the definition in the particular case $\varepsilon = 1$, and then use Remark \ref{eps-normalization-remark} to generalize it to an arbitrary $\varepsilon$.

\begin{definition}
Fix real $w$ and $d$ called \emph{dip width} and \emph{depth exponents} respectively. For a positive integer $T$, a point $v \in \left(-\frac{1}{\sqrt{1 + m^2}}; \frac{1}{\sqrt{1 + m^2}}\right)$ is called a \textit{dip of order $T$} of the function $\tilde a_1(x, t, m)$, if and only if for each integer sequence $x_t$ satisfying
\begin{align}
\label{dips-seq-estimate}
	x_t = vt + o(t^w) \qquad\text{as $t \to \infty$},
\end{align}
we have
\begin{align}
\label{dips-lim-condition}
	\tilde a_1(x_t, t, m) - \tilde a_1(x_t - 2T, t, m) = o(t^d) \qquad\text{as $t \to \infty$}.
\end{align}
The point $v \in \left(-\frac{1}{\sqrt{1 + m^2\varepsilon^2}}; \frac{1}{\sqrt{1 + m^2\varepsilon^2}}\right)$ is called a \emph{dip of order $T$} of the function $\tilde a_1(x, t, m, \varepsilon)$, if and only if it is a dip of order $T$ of the function $\tilde a_1(x, t, m\varepsilon)$.
\end{definition}

In other words, we highlight all $v$ such that for every fixed $t$ and all $x$ giving the same remainder modulo $2T\varepsilon$ the oscillation of the function $\tilde a_1(x, t, m, \varepsilon)$ near $x = vt$ is small enough. Informally, the points $([vt], \tilde a_1([vt], t))$ stand out on Figures \ref{layer-a1-figure}--\ref{dips-a1-figure}, because near the local maximum and minimum the oscillation is less than in the vicinity, that is why the graph is sparse there. Proposition \ref{asymptotic-a1-proposition} implies that the oscillation amplitude of the function $\tilde a_1(x, t, m, \varepsilon)$ is of order $1/\sqrt{t}$, thus we should take $d = -1/2$. Numerical experiments show that the dip width is of order $\sqrt{t}$, thus we take $w = 1/2$.

\begin{theorem}
\label{dips-a1-theorem}
Assume $w = 1/2$ and $d = -1/2$. If $m > 0$ and $\varepsilon \le 1/m$, then all dips of order $T$ of the function $\tilde a_1(x, t, m, \varepsilon)$ are exactly the points
\begin{align}
\label{dips-v-formula}
	v = \frac{\sin(\pi k/T)}{\sqrt{m^2\varepsilon^2 + \sin^2(\pi k/T)}} \quad\text{for $k \in \mathbb{Z} \cap \left(-\frac{T}{2}; \frac{T}{2}\right)$}.
\end{align}
\end{theorem}

\begin{remark}
Formula \eqref{dips-v-formula} has a curious physical interpretation. Consider $\omega_p := \frac{1}{\varepsilon} \arccos{\frac{\cos{p\varepsilon}}{\sqrt{1 + m^2\varepsilon^2}}}$, which is equal to the energy of an electron with momentum $p$ in Feynman checkers (see \cite[after Proposition 12]{SU-22}). The right side of \eqref{dips-v-formula} is obviously $\frac{\partial \omega_p}{\partial p} \rvert_{p = \frac{\pi k}{T\varepsilon}}$, which corresponds to the velocity of an electron with momentum $p$ due to the Hamilton--Jacobi equation $\frac{\partial \omega_p}{\partial p} = \frac{dx}{dt}$. Since $\hbar = 1$ in our units, the de Broglie wavelength is $2\pi/p$. Thus the dips occur for those electron velocities which correspond to de Broglie wavelengths being rational multiplies of $\varepsilon$. In other words, the dips are explained by electron diffraction on the lattice $\varepsilon \mathbb{Z}^2$.
\end{remark}

In addition to the physical meaning, this result allows us to prove the sharpness of the lower bound in Theorem \ref{oscillation-outside-theorem}.

\begin{proposition}
\label{lowerbound-sharpness-proposition}
If $m > 0$ and $\varepsilon \le 1/m$, then for every real $v_0 < \frac{1}{\sqrt{1 + m^2\varepsilon^2}}$ there exists $(x, t) \in \varepsilon \mathbb{Z}^2$ such that
\begin{align*}
	t > 0, \qquad \frac{x + t}{\varepsilon} \mathrel{\not\vdots} 2, \qquad v_0 \le \frac{x}{t} < \frac{1}{\sqrt{1 + m^2\varepsilon^2}}, \qquad \mathrm{sgn}(\tilde a_1(x, t, m, \varepsilon)) \neq (-1)^{\frac{t - x - \varepsilon}{2\varepsilon}}.
\end{align*}
\end{proposition}

Let us introduce some notation to be used in the proofs below. For each positive $p \in \mathbb{R}$ and each $a, b \in \mathbb{R} / p \mathbb{Z}$ denote by $\rho_p(a, b)$ the distance from $b - a$ to the closest point of $p \mathbb{Z}$. For a sequence of real numbers $x_n$ and $c \in \mathbb{R}$, the notation ``$x_n \xrightarrow{p} c$'' means the convergence of $x_n$ to $c$ in the metric $\rho_p$, in other words, the convergence $x_n \bmod p \mathbb{Z} \to c \bmod p \mathbb{Z}$ in $\mathbb{R} / p \mathbb{Z}$. Set $\varepsilon = 1$ without loss of generality; see Remark \ref{eps-normalization-remark}. Fix arbitrary $m > 0$, $T \ge 1$, $v \in \left(-\frac{1}{\sqrt{1 + m^2}}; \frac{1}{\sqrt{1 + m^2}}\right)$, a sequence $x_t$ fulfilling \eqref{dips-seq-estimate}, and denote
\begin{align*}
	x'_t &:= x_t - 2T, \\
	\delta_t^- = \delta_t^-(x_t) &:= \theta(x_t, t, m) - \theta(x'_t, t, m), \\
	\delta_t^+ = \delta_t^+(x_t) &:= \theta(x_t, t, m) + \theta(x'_t, t, m), \\
	\gamma_t = \gamma_t(x_t) &:= \min \left\{\rho_{2\pi}(0, \delta_t^-), \rho_{2\pi}(\pi, \delta_t^+)\right\}.
\end{align*}
Here we use notation \eqref{theta-definition}. Note that there is $t_0$ such that for every $t > t_0$ the values $\theta(x_t, t, m)$ and $\theta(x'_t, t, m)$ are well-defined, because $\theta(x, t, m)$ is well-defined for all $x, t, m$ satisfying $|x/t| < 1/\sqrt{1 + m^2}$. Hereafter we assume that $t > t_0$ and omit the phrase ``as $t \to \infty$''. Take $\delta_v > 0$ such that $|v| < 1/\sqrt{1 + m^2} - \delta_v$.

\begin{lemma}
\label{dips-gamma-lemma}
Assume $w = 1/2$, $d = -1/2$, $0 < m \le 1$ and \eqref{dips-seq-estimate}. Then condition \eqref{dips-lim-condition} is equivalent to the condition $\gamma_t \to 0$.
\end{lemma}

\begin{proof}
Since $x_t/t \to v$, by Proposition \ref{asymptotic-a1-proposition} for large enough $t$ we have
\begin{align*}
	t^{1/2}(\tilde a_1(x_t, t, m) - \tilde a_1(x'_t, t, m)) = \sqrt{\frac{2m}{\pi}} ((1 - (1 + m^2)(x_t/t)^2)^{-1/4} \sin \theta(x_t, t, m)\,- \\
	-\,(1 - (1 + m^2)(x'_t/t)^2)^{-1/4} \sin \theta(x'_t, t, m)) + O_{\delta_v}\left(\frac{1}{m^{1/2}t}\right) = \\
	= \sqrt{\frac{2m}{\pi}} (1 - (1 + m^2)v^2)^{-1/4} (\sin \theta(x_t, t, m) - \sin \theta(x'_t, t, m)) + o(1).
\end{align*}
Thus, $v$ is a dip if and only if
\begin{align*}
	\sin \theta(x_t, t, m) - \sin \theta(x'_t, t, m) \to 0,
\end{align*}
which is equivalent to $\gamma_t \to 0$.
\end{proof}

\begin{lemma}
\label{dips-delta-minus-lemma}
Assume \eqref{dips-seq-estimate} for $w = 1/2$ and $m > 0$. Then the sequence $\delta_t^-$ converges. Moreover, the limit is 0 in the metric $\rho_{2\pi}$ if and only if $v$ fulfills \eqref{dips-v-formula} for $\varepsilon = 1$.
\end{lemma}

\begin{proof}
We can rewrite $\theta(x, t, m)$ in the following form:
\begin{align*}
    \theta(x, t, m) = t L(x/t) + \frac{\pi}{4} \quad\text{for some}\quad L(v) \in C^2\left(-\frac{1}{\sqrt{1 + m^2}}; \frac{1}{\sqrt{1 + m^2}}\right).
\end{align*}
Write the Taylor expansion of $L(x/t)$ at the point $v$ for $\left|x/t\right| < 1/\sqrt{1 + m^2} - \delta_v$:
\begin{align}
\label{dips-theta-expansion}
	\theta(x, t, m) = \theta(vt, t, m) - t \left(\frac{x}{t} - v\right) \arcsin \frac{mv}{\sqrt{1 - v^2}} + O_{\delta_v, m}\left(t \left(\frac{x}{t} - v\right)^2\right).
\end{align}
Using the estimate $\frac{x_t}{t} - v = o\left(\frac{1}{\sqrt{t}}\right) = \frac{x'_t}{t} - v$ from \eqref{dips-seq-estimate} we get
\begin{align*}
	\delta_t^- = (x'_t - vt) \arcsin \frac{mv}{\sqrt{1 - v^2}} - (x_t - vt) \arcsin \frac{mv}{\sqrt{1 - v^2}} + o(1) = (x'_t - x_t) \arcsin \frac{mv}{\sqrt{1 - v^2}} + o(1) = \\
	= -2T \arcsin \frac{mv}{\sqrt{1 - v^2}} + o(1).
\end{align*}
Thus $\delta_t^- \xrightarrow{2\pi} 0$ if and only if
\begin{align*}
	T \arcsin \frac{mv}{\sqrt{1 - v^2}} = \pi k \quad\text{for some $k \in \mathbb{Z}$},
\end{align*}
which is equivalent to \eqref{dips-v-formula} for $\varepsilon = 1$, assuming that $v \in \left(-\frac{1}{\sqrt{1 + m^2}}; \frac{1}{\sqrt{1 + m^2}}\right)$.
\end{proof}

This proof has a clear physical meaning. The function $L(v)$ is the Lagrangian \cite[end of \S 2.4]{SU2-22}. The arcsine in \eqref{dips-theta-expansion} equals $\partial L / \partial v$, that is, the momentum. Thus the dips indeed occur for those electron velocities which correspond to the momenta $\pi k / T\varepsilon$. 

\begin{lemma}
\label{dips-delta-plus-lemma}
For every integer $T \ge 1$, real $m > 0$, $v \in \left(-\frac{1}{\sqrt{1 + m^2}}; \frac{1}{\sqrt{1 + m^2}}\right)$, and $c \in \mathbb{R} / 2\pi \mathbb{Z}$ there is an integer sequence $x_t$ fulfilling \eqref{dips-seq-estimate} for $w = 1/2$, such that $\delta_t^+ \not\xrightarrow{2\pi} c$.
\end{lemma}

\begin{proof}
Assume the converse: there is $c \in \mathbb{R} / 2\pi \mathbb{Z}$ such that $\delta_t^+ \xrightarrow{2\pi} c$ for every integer sequence $x_t$ fulfilling \eqref{dips-seq-estimate} for $w = 1/2$. Using expansion \eqref{dips-theta-expansion} and estimate \eqref{dips-seq-estimate} we get
\begin{align}
\label{dips-delta-plus-expansion}
	\delta_t^+ = 2\theta(vt, t, m) + (2vt - x_t - x'_t) \arcsin \frac{mv}{\sqrt{1 - v^2}} + o(1).
\end{align}
Take integer sequences $y_t$ and $z_t$ fulfilling \eqref{dips-seq-estimate}, such that $z_t - y_t = 2$ for all $t$. Due to our assumption,
\begin{align*}
	\delta_t^+(y_t) - \delta_t^+(z_t) \xrightarrow{2\pi} c - c = 0.
\end{align*}
On the other hand, from expansion \eqref{dips-delta-plus-expansion} it follows that
\begin{align*}
	\delta_t^+(y_t) - \delta_t^+(z_t) = 2(z_t - y_t) \arcsin \frac{mv}{\sqrt{1 - v^2}} + o(1) = 4 \arcsin \frac{mv}{\sqrt{1 - v^2}} + o(1).
\end{align*}
Consequently,
\begin{align*}
	\arcsin \frac{mv}{\sqrt{1 - v^2}} = \frac{\pi k}{2} \quad\text{for some $k \in \mathbb{Z}$}.
\end{align*}
Since $v \in \left(-\frac{1}{\sqrt{1 + m^2}}; \frac{1}{\sqrt{1 + m^2}}\right)$, the latter equality clearly holds only for $v = 0$. In the latter case, applying expansion \eqref{dips-delta-plus-expansion} we get
\begin{align*}
	\delta_t^+ = 2\left(t \arcsin \frac{m}{\sqrt{1 + m^2}} + \frac{\pi}{4}\right) + o(1) = \frac{\pi}{2} + 2t \arcsin \frac{m}{\sqrt{1 + m^2}} + o(1) \not\xrightarrow{2\pi} c,
\end{align*}
because $\arcsin \frac{m}{\sqrt{1 + m^2}} \neq \pi k$ for any $k \in \mathbb{Z}$, a contradiction.
\end{proof}

\begin{proof}[Proof of Theorem \ref{dips-a1-theorem}]
Set $\varepsilon = 1$ without loss of generality. If $v$ fulfills \eqref{dips-v-formula}, then by Lemma \ref{dips-delta-minus-lemma} we have
\begin{align*}
	\gamma_t \le \rho_{2\pi}(0, \delta_t^-) \to 0,
\end{align*}
thus $v$ is a dip by Lemma \ref{dips-gamma-lemma}. If $v$ is not described by formula \eqref{dips-v-formula}, then by Lemmas \ref{dips-delta-minus-lemma} and \ref{dips-delta-plus-lemma} applied for $c = \pi$ there is $a > 0$ such that
\begin{align*}
	\rho_{2\pi}(0, \delta_t^-) \to a \quad\text{and}\quad \rho_{2\pi}(\pi, \delta_t^+) \not\to 0,
\end{align*}
thus $\gamma_t \not\to 0$. Therefore, $v$ is not a dip by Lemma \ref{dips-gamma-lemma}.
\end{proof}

\begin{proof}[Proof of Proposition \ref{lowerbound-sharpness-proposition}]
Without loss of generality set $\varepsilon = 1$. By Theorem \ref{dips-a1-theorem} and formula \eqref{dips-v-formula} it follows that the set of dips of an odd order is dense on the interval $\left(-\frac{1}{\sqrt{1 + m^2}}; \frac{1}{\sqrt{1 + m^2}}\right)$. Therefore, there is a dip $v > v_0$ of order $T \mathrel{\not\vdots} 2$.

By Lemma \ref{dips-delta-plus-lemma} for $c = 0$ there is a sequence $x_t = vt + o(t^{1/2})$ such that $\delta_t^+ \not\xrightarrow{2\pi} 0$. By Lemma \ref{dips-delta-minus-lemma} we have $\delta_t^- \xrightarrow{2\pi} 0$. Adding the two sequences, we get $\delta_t^+ + \delta_t^- = 2\theta(x_t, t, m) \not\xrightarrow{2\pi} 0$, meaning that $\theta(x_t, t, m) \not\xrightarrow{\pi} 0$; hence $\sin \theta(x_t, t, m) \not\to 0$. By Proposition \ref{asymptotic-a1-proposition} it follows that
\begin{align*}
	\tilde a_1(x_t, t, m) \neq o(t^{-1/2}).
\end{align*}
Since $v$ is a dip of order $T$,
\begin{align*}
	\tilde a_1(x_t, t, m) - \tilde a_1(x_t - 2T, t, m) = o(t^{-1/2}).
\end{align*}
By the latter two estimates it follows that there is an arbitrarily large $t$ such that
\begin{align*}
	\mathrm{sgn}(\tilde a_1(x_t, t, m)) = \mathrm{sgn}(\tilde a_1(x_t - 2T, t, m)).
\end{align*}
Therefore, since $T$ is odd, either $x = x_t$ or $x = x_t - 2T$ satisfies the required conditions.
\end{proof}

\section*{Open problems}

We conjecture that the functions $\tilde a_2(x, t, m, \varepsilon)$ and $P(x, t, m, \varepsilon)$ have the same positions of the dips as $\tilde a_1(x, t, m, \varepsilon)$ (with the depth exponent $-1$ for the case of $P(x, t, m, \varepsilon)$). The proof should be analogous, using the asymptotic formulae \cite[Theorem 2]{SU-22} and \cite[Theorem 1]{Sunada-Tate-12} although more technical.

It is interesting to generalize Theorem \ref{oscillation-middle-theorem} to arbitrary $m$ and $\varepsilon$ and find an explicit bound of $w_0$ in Theorem \ref{young-middle-theorem}. Note that we have proved that for each fixed $x$ and all sufficiently large $t$, the value $\tilde a_1(x, t)$ is not zero (Corollary \ref{a1-nonzero-corollary}); the problem to find all zeros of the function $\tilde a_1(x, t)$ remains open; cf. \cite[Theorem 1]{Novikov-22}. We finish the work by the following conjecture.

\begin{conjecture}
Let $x$, $t$ be integers, $t \ge 1$, $-t + 2 < x < t$, $x + t \mathrel{\vdots} 2$. Then $\operatorname{Re} a(x, t) \neq 0$ and $\operatorname{Im} a(x, t) \neq 0$ unless $(x, t) \in \{(-3, 11), (5, 11)\}$ or $x \in \{0, 2\}$.
\end{conjecture}

\section*{Acknowledgements}

This paper was prepared within the framework of the HSE University Basic Research Program and supported by project group ``Lattice models'' in Faculty of Mathematics. We are grateful to Mikhail Skopenkov for inspiration and organization of this work and to Ivan Novikov for useful discussions.


\begin{thebibliography}{DoSn}
	\bibitem{Ambainis-etal-01} A.~Ambainis, E.~Bach, A.~Nayak, A.~Vishwanath, J.~Watrous, One-dimensional quantum walks, Proc. of the 33rd Annual ACM Symposium on Theory of Computing (2001), 37--49.
	\bibitem{Dmitriev-22} M.~Dmitriev, Feynman checkers with absorption, preprint (2022), \href{https://arxiv.org/abs/2204.07861}{arXiv:2204.07861}.
	\bibitem{Feynman-Gibbs} R.P.~Feynman,  A.R.~Hibbs, Quantum mechanics and path integrals, New York, McGraw-Hill, 1965. Russian transl.: Mir, Moscow, 1968.
	\bibitem{Nechaev-Polovnikov-18} S. K. Nechaev, K. Polovnikov, “Rare-event statistics and modular invariance”, UFN, 188:1 (2018), 106–112; Phys. Usp., 61:1 (2018), 99–104.
	\bibitem{Novikov-22} I.~Novikov, Feynman checkers: the probability to find an electron vanishes nowhere inside the light cone, Reviews Math Physics 34:07, 2250020 (2022).
	\bibitem{Ozhegov-22} F.~Ozhegov, Feynman checkers: external electromagnetic field and asymptotic properties, preprint (2022), \href{ https://arxiv.org/abs/2209.00938}{arXiv:2209.00938}.
	\bibitem{SU-22} M.~Skopenkov, A.~Ustinov, Feynman checkers: towards algorithmic quantum theory, Russian Math. Surveys 77:3(465) (2022), 73-160.
	\bibitem{SU2-22} M.~Skopenkov, A.~Ustinov, Feynman checkers: Minkowskian lattice field theory, preprint (2022), \href{https://arxiv.org/abs/2208.14247}{arXiv:2208.14247}.
	\bibitem{Sunada-Tate-12} T.~Sunada, T.~Tate, Asymptotic behavior of quantum walks on the line, J. Funct. Anal. 262 (2012) 2608--2645.
	\bibitem{Zakorko-22} P.~Zakorko, Feynman checkers: a uniform approximation of the wave function by Airy function, preprint (2022).
\end{thebibliography}
\end{document}